\documentclass[sigconf ]{acmart}
\usepackage{enumitem}
\usepackage{amsmath}
  
\usepackage{tcolorbox}
\usepackage{amssymb}
\usepackage{bm}
\usepackage{nicefrac}
\usepackage{booktabs}
\usepackage{array}
\usepackage{multirow}
\usepackage{threeparttable}
\usepackage{makecell}
\usepackage[procnumbered,ruled,vlined,linesnumbered]{algorithm2e}
\usepackage{siunitx}
\usepackage{stfloats}
\usepackage{graphicx}
\usepackage{subfigure}
\usepackage{hyperref}
\usepackage{enumerate}
\usepackage{graphicx}
\usepackage{caption}
\usepackage{subcaption}

\newtheorem{theorem}{Theorem}[section]

\newtheorem{lemma}[theorem]{Lemma}

\newtheorem{definition}[theorem]{Definition}

\newenvironment{fminipage}%
{\begin{Sbox}\begin{minipage}}%
		{\end{minipage}\end{Sbox}\fbox{\TheSbox}}

\def\calG{\mathcal{G}}

\def\calF{\mathcal{F}}
\def\calR{\mathcal{R}}
\newcommand{\removelatexerror}{\let\@latex@error\@gobble}

\newcommand\LL{\bm{\mathit{L}}}

\newcommand\ee{\boldsymbol{\mathit{e}}}

\renewcommand\AA{\boldsymbol{\mathit{A}}}

\newcommand\DD{\boldsymbol{\mathit{D}}}

\newcommand\II{\boldsymbol{\mathit{I}}}

\DontPrintSemicolon
\SetKw{KwAnd}{and}
\SetFuncSty{textsc}
\SetKwInOut{Input}{Input\ \ \ \ }
\SetKwInOut{Output}{Output}

\usepackage{tabularx}
\usepackage{stfloats}

\DontPrintSemicolon
\SetKw{KwAnd}{and}
\SetFuncSty{textsc}
\SetKwInOut{Input}{Input\ \ \ \ }
\SetKwInOut{Output}{Output}

\usepackage{tabularx}
\usepackage{stfloats}

\copyrightyear{2024}
\acmYear{2024}
\setcopyright{acmlicensed}
\acmConference[KDD '24] {Proceedings of the 30th ACM SIGKDD Conference on Knowledge Discovery and Data Mining }{August 25--29, 2024}{Barcelona, Spain.}
\acmBooktitle{Proceedings of the 30th ACM SIGKDD Conference on Knowledge Discovery and Data Mining (KDD '24), August 25--29, 2024, Barcelona, Spain}
\acmDOI{10.1145/3637528.3671822}
\acmISBN{979-8-4007-0490-1/24/08}

\begin{document}
%	\fancyhead{}
	\title{Fast Computation for the Forest Matrix of an Evolving Graph} 
%\footnotemark[1]
% \author{Haoxin Sun and Zhongzhi Zhang\footnotemark}
% \affiliation{
% \institution{Shanghai Key Laboratory of Intelligent Information Processing, Fudan University, Shanghai 200433, China
% \\School of Computer Science, Fudan University, Shanghai 200433, China}
% 		\city{}
% 		\country{}
% }
% \email{{21210240097,zhangzz}@fudan.edu.cn}

	\author{Haoxin Sun}
	\affiliation{%
		\institution{Fudan University}
		%\streetaddress{1 Th{\o}rv{\"a}ld Circle}
		\city{Shanghai}
		\country{China}}
	\email{21210240097@m.fudan.edu.cn}

 	\author{Xiaotian Zhou}
	\affiliation{%
		\institution{Fudan University}
		%\streetaddress{1 Th{\o}rv{\"a}ld Circle}
		\city{Shanghai}
		\country{China}}
	\email{22110240080@m.fudan.edu.cn}
 
	\author{Zhongzhi Zhang\footnotemark}
	\affiliation{%
		\institution{Fudan University}
		%\streetaddress{1 Th{\o}rv{\"a}ld Circle}
		\city{Shanghai}
		\country{China}}
	\email{zhangzz@fudan.edu.cn}

\begin{abstract}

% The forest matrix plays a crucial role in network science, opinion dynamics, and machine learning, offering deep insights into the structure and dynamics of networks.  Given that the real-world graphs are usually dynamically changing, current algorithms face efficiency challenges when applied to evolving graphs. In this paper,  we study the problem of querying entries  of the forest matrix in evolving graphs. We first   propose algorithm \textsc{SFQ}, which employs a probabilistic interpretation of the forest matrix. Building on this,  we develop the \textsc{SFQPlus} algorithm, which incorporates a novel variance reduction technique and is theoretically proven to offer enhanced accuracy. To adapt our algorithms to evolving graphs characterized by edge additions and deletions, we devise a strategy that effectively maintains a list of spanning converging forests, ensuring unbiased estimates of the forest matrix with an $O(1)$ runtime complexity for both updates and queries.  Finally, thorough   extensive experiments on various real-world networks, we demonstrate the efficiency and effectiveness of our algorithms.  Particularly, our algorithms are scalable to massive graphs with more than forty million nodes. 

The forest matrix plays a crucial role in network science, opinion dynamics, and machine learning, offering deep insights into the structure of and dynamics on networks. In this paper, we study the problem of querying entries of the forest matrix in evolving graphs, which more accurately represent the dynamic nature of real-world networks compared to static graphs.   To address the unique challenges posed by evolving graphs, we first introduce two approximation algorithms,  \textsc{SFQ} and \textsc{SFQPlus},  for static graphs.   \textsc{SFQ} employs a probabilistic interpretation of the forest matrix, while  \textsc{SFQPlus} incorporates a novel variance reduction technique and is theoretically proven to offer enhanced accuracy. Based on these two algorithms, we further devise two dynamic algorithms centered around efficiently maintaining a list of spanning converging forests. This approach ensures $O(1)$ runtime complexity for updates, including edge additions and deletions, as well as for querying matrix elements, and provides an unbiased estimation of forest matrix entries.  Finally, through extensive experiments on various real-world networks, we demonstrate the efficiency and effectiveness of our algorithms. Particularly, our algorithms are scalable to massive graphs with more than forty million nodes.

\end{abstract}
	
	%%the proposed problem is NP-hard, but prove


% 	<ccs2012>
% 	<concept>
% 	<concept_id>10003752.10003809.10003635</concept_id>
% 	<concept_desc>Theory of computation~Graph algorithms analysis</concept_desc>
% 	<concept_significance>500</concept_significance>
% 	</concept>
% 	<concept>
% 	<concept_id>10002950.10003624.10003625.10003626</concept_id>
% 	<concept_desc>Mathematics of computing~Combinatoric problems</concept_desc>
% 	<concept_significance>500</concept_significance>
% 	</concept>
% 	<concept>
% 	<concept_id>10003033.10003068</concept_id>
% 	<concept_desc>Networks~Network algorithms</concept_desc>
% 	<concept_significance>500</concept_significance>
% 	</concept>
% 	</ccs2012>

% \ccsdesc[500]{Theory of computation~Graph algorithms analysis}
% \ccsdesc[500]{Mathematics of computing~Combinatoric problems}
% \ccsdesc[500]{Networks~Network algorithms}
%	\begin{CCSXML}
%	
%	\end{CCSXML}
%	
%	\ccsdesc[500]{Theory of computation~Graph algorithms analysis}
%	\ccsdesc[500]{Theory of computation~Discrete optimization}
%	\ccsdesc[500]{Information systems~Data mining}
	

	%% A "teaser" image aNeNepears between the author and affiliation
%	\begin{teaserfigure}
%		\includegraphics[width=\textwidth]{sampleteaser}
%		\caption{Seattle Mariners at Spring Training, 2010.}
%		\Description{Enjoying the baseball game from the third-base
%			seats. Ichiro Suzuki preparing to bat.}
%		\label{fig:teaser}
%	\end{teaserfigure}
	

\maketitle
%\renewcommand{\thefootnote}{\fnsymbol{footnote}}
\renewcommand{\thefootnote}{*}
%\footnotetext[1]{Corresponding author. Zhongzhi Zhang is also with Shanghai Blockchain Engineering Research Center, as well as Research Institute of Intelligent Complex Systems, Fudan University, Shanghai 200433.}
\renewcommand{\thefootnote}{*}
\footnotetext[1]{Corresponding author. Haoxin Sun, Xiaotian Zhou,  and Zhongzhi Zhang are with  Shanghai Key Laboratory of Intelligent Information Processing, School of Computer Science, Fudan University, Shanghai 200433, China. } %Zhongzhi zhang is also with Institute of Intelligent Complex Systems, Fudan University, Shanghai 200433, China.

\section{Introduction}

 As a fundamental representation of a graph, the Laplacian matrix  $\LL$ encapsulates  a wealth of structural and dynamical information of the graph~\cite{Me94}. The forest matrix, denoted as $\mathbf{\Omega} = (\II + \LL)^{-1}$, also plays a pivotal role in the field of network science. This matrix, along with its variants, is recognized as a foundational matrix in numerous applications across various domains, such as Markov processes~\cite{AvLuGaAl18,AvCaGaMe18}, opinion dynamics~\cite{GiTeTs13, XuBaZh21,SuZh23}, graph signal processing~\cite{PiAmBaTr21,PiAmBaTr20}, game theory~\cite{BaCaZe06,GaGoGo20}, and multi-label classification~\cite{DeZhLiCh17}. The entries of the forest matrix notably reflect the structural properties of the graph due to their close relationship with the spanning rooted forests within the graph~\cite{ChSh06,ChSh97}.   In particular, the diagonal entries of $\mathbf{\Omega}$ are associated with the concept of forest closeness centrality~\cite{JiBaZh19, GrAnPrMe21} in graph mining and determinantal point processes~\cite{KuTa12} in machine learning, and find applications in electrical interpretations for multi-agent and network-based challenges~\cite{RoFrFa17}. The off-diagonal elements, $\omega_{ij}$, serve as a measure of the proximity of  node $i$ and node $j$, with a lower $\omega_{ij}$ indicating a greater ``distance'' between the two nodes~\cite{ChSh97}. These elements are also instrumental in calculating the personalized PageRank centrality between two nodes~\cite{JeWi03,WaYaXiWeYa17}.

Due to the broad applications of the forest matrix, querying its elements is of significant importance. In order to query the entries of $\mathbf{\Omega}$ for a graph with $n$ nodes, a direct inversion of the matrix $\II+\LL$   costs $O(n^3)$ operations and $O(n^2)$ memories and thus is prohibitive for relatively large graphs. In previous work, some fast Laplacian solver~\cite{CoKyMiPaJaPeRaXu14} based algorithms were proposed to compute the diagonal of the forest matrix~\cite{JiBaZh19, GrAnPrMe21}. Besides, some forest sampling algorithms were proposed to compute the trace of the forest matrix~\cite{BaTrGaAvPO19,PiAmBaTr22trace} and the column sum of the forest matrix~\cite{SuZh23}. 

A majority of existing algorithms are  designed for static graphs, despite the fact that many real-life networks typically evolve dynamically.  To query the elements of the forest matrix in evolving graphs, we need to repeatedly run these algorithms as the graph structure changes, which is very inefficient. Recognizing this limitation,  there is a pressing need for a dynamic approach to querying elements of the forest matrix in evolving graphs. Such a dynamic solution should be able to be updated  easily whenever edges are added or removed from the network, offering query results significantly faster than recalculating from the beginning. Furthermore, it is equally crucial to ensure that  the solution has guaranteed accuracy.

% However, in social networks, their structures often change, and these algorithms are designed for static graphs and run ineffectively on evolving graphs.   Consequently, a theoretically guaranteed estimation algorithm for approximating the entry of $\mathbf{\Omega}$ for both static and evolving  graphs is imperative.
 
In this paper, we delve deep into the problem of efficiently computing  the entries of the forest matrix in an evolving  digraph, in order to overcome the challenges and limitations of existing algorithms. The main contributions of this work are summarized as follows:

(i) For static graphs, we introduce an algorithm \textsc{SFQ} to query forest matrix entries through a probabilistic interpretation. To enhance accuracy and reduce variance, we develop innovative variance-reduction techniques and present an algorithm \textsc{SFQPlus}, alongside a theoretical guarantee for its performance.

(ii) In the context of evolving graphs, we focus on two edge operations: edge insertions and deletions, by devising an update strategy to maintain a list of spanning converging forests. We demonstrate that our algorithm provides an unbiased estimate of the forest matrix entries and maintains an $O(1)$ runtime complexity for both querying and updating.

(iii) Comprehensive experiments on diverse real-world undirected and directed networks show that \textsc{SFQPlus} consistently delivers superior estimation accuracy compared to \textsc{SFQ}, which returns results close to the ground truth. Furthermore, our algorithms are efficient and scalable, even on massive graphs with more than forty million nodes.

% In this paper, we introduce three novel sampling-based algorithms, namely \textsc{SFQ}, \textsc{SFQV}, and \textsc{SFQV+}, to  compute the diagonal of the forest matrix  
%  $\mathbf{\Omega}$  of a digraph $\calG = (V,E)$ with $n$ nodes and $m$ edges. Based on an extension of Wilson's algorithm~\cite{Wi96,WiPr96}, we first propose \textsc{SFQ}, which returns  an approximation of the actual diagonal entries of  $\mathbf{\Omega}$ with time complexity of $O(ln)$, where $l$ represents  the sampling number. To reduce the variance in sampling and enhance accuracy, we subsequently propose two additional algorithms, \textsc{SFQV} and \textsc{SFQV+}, each offering progressively superior accuracy and progressively reduced variance. Notably, \textsc{SFQV+} achieves an $\epsilon$ relative error guarantee for each node at a time complexity of $O(n\log(1/\delta)/\epsilon^2)$, with a success probability of $1-\delta$.  Finally, we perform extensive experiments on both undirected and directed real-world networks, which shows that compared to the baseline approaches, our algorithms are both effective and efficient. Furthermore, our algorithms scale to massive graphs with more than thirty million nodes.
 % \vspace{-0.3cm}
\section{Related Work}
   In this section, we briefly review the existing work related to ours.

The investigation into forest matrix entries has garnered significant attention in recent years, leading to the development of two main categories of algorithms for addressing related problems: solver-based algorithms and sampling-based algorithms.  Solver-based algorithms primarily operate on undirected graphs, leveraging the fast Laplacian solver~\cite{CoKyMiPaJaPeRaXu14}. In works such as~\cite{JiBaZh19,GrAnPrMe21}, the authors utilized the fast Laplacian solver for fast calculation of the diagonal of forest matrix. In~\cite{XuBaZh21}, the authors combined the Johnson-Lindenstrauss lemma \cite{JoLi84,Ac03} with the fast Laplacian solver to compute relevant quantities related to the forest matrix. Another application in~\cite{ZhZh22} utilizes a similar method for addressing optimization issues in social dynamics, fundamentally relying on the forest matrix.

In addition to solver-based algorithms, many algorithms are sampling-based, inspired by the matrix-forest theorem that establishes a link between spanning forests and the forest matrix~\cite{ChSh06,ChSh97}.  Wilson's algorithm plays a pivotal role in these sampling-based algorithms. Wilson's algorithm and its variants have found applications across various domains, such as computing the trace of the forest matrix~\cite{BaTrGaAvPO19,PiAmBaTr22trace}, estimating the diagonal of forest matrix~\cite{SuZh24}, computing the column sum of the forest matrix to solve optimization problems in opinion dynamics~\cite{SuZh23}, and solving linear systems in graph signal processing related to the forest matrix~\cite{PiAmBaTr21,PiAmBaTr20}.

For many realistic  large graphs like social networks and the Internet, their structures often change over time, posing challenges in maintaining up-to-date information. For example, as the graph's structure evolves, it becomes necessary to repeatedly run previously mentioned algorithms, whether solver-based or sampling-based, to obtain newly queried forest matrix data. Researchers have developed many special tools called dynamic graph algorithms, which help solve problems faster. These tools have been used for various problems that involve edge sparsifiers~\cite{ GoHeTh18, HoLiTh01, KaKiMo13, Th07}, as well important variants of edge sparsifiers themselves, including minimum spanning trees~\cite{ HoRoWu15, NaSa17, NaSaWu17, Wu17}, spanners~\cite{ BaKhSa12}, spectral sparsifiers~\cite{ AbDuKoKrPe16, DuGaGoPe19, BrGaJaLeLiPeSi22}, and low-stretch spanning trees~\cite{ FoGo19}. However, most of these advancements have been more theoretical than practical and may not be suitable for the forest matrix query problem.

Our work takes a step further by applying Wilson's algorithm and providing novel techniques to reduce variance, for quickly updating and querying any entry in the forest matrix in a graph as it changes.

\section{Preliminaries}

 In this section, we introduce some useful notations and tools for the convenience of description and analysis.
% \subsection{Notations}

% We use normal lowercase letters like $ a,b,c $ to denote scalars in set of real numbers, normal uppercase letters like $ A,B,C $ to denote sets, bold lowercase letters like $ \aaa,  \bb, \cc$ to denote column vectors, and bold uppercase letters like $ \AA,\BB,\CC $ to denote matrices.  We use $\AA^{\top}$ and  $\aaa^{\top}$ to represent the transpose of $\AA$ and  $\aaa$, respectively. Let $ \ee_i $  denote the column vector of appropriate dimension, where the $ i $-th element is $ 1 $, and other elements are $ 0 $. Let $ \mathbf{0} $ be an appropriate-dimension column vector with all entries being zeros, and let $\mathbf{1}$ be an appropriate-dimension column vector with all entries being ones. Let $ \II $ denote an appropriate-dimension identity matrix.  For a matrix $\AA$, $\AA_{H,F} $ denotes the submatrix of $\AA $ with row indices in set $H$ and column indices in set $ F$, and $ A_{-H} $ denotes the submatrix of $ \AA $ obtained from $ \AA $ by deleting rows and columns corresponding to nodes in set $ H $. Similarly, for a vector $ \aaa $, we use $ \aaa_{-H} $ to denote the vector obtained from $ \aaa $ by deleting elements in set $H$. If $ H $ contains only a single element $i$, we use $ \AA_{-i}$ and $\aaa_{-i}$ to denote, respectively, $ \AA_{-\{i\}} $ and $ \aaa_{-\{i\}} $ for simplicity.

\subsection{ Graph and  Laplacian Matrix}
Let $\calG=(V,E)$ denote  an unweighted simple directed graph (digraph), which consists of $n=|V|$ nodes (vertices) and $m=|E|$ directed edges (arcs). Here, \(V=\{v_1,v_2,\ldots,v_n\}\) denotes the set of nodes, and \(E\) represents the set of directed edges.  An directed edge \( (v_i,v_j) \in E \) indicates an edge pointing from node \(v_i\) to node \(v_j\). In what follows, $v_i$ and $i$ are used interchangeably to represent node $v_i$ if incurring no confusion.    

% The structure information of digraph $\calG=(V,E)$  is characterized by  its adjacency matrix $\AA=(a_{ij})_{n\times n}$, where $a_{ij} = 1$  if $ (v_i,v_j) \in E $ and $a_{ij} = 0 $ otherwise. For any node \(i\) in \(\calG\), its in-degree \(d^+_i\) and out-degree \(d^-_i\) are given by \(d^+_i=\sum_{j=1}^n a_{ji}\) and \(d^-_i=\sum_{j=1}^n a_{ij}\), respectively. In the sequel, we use $d_i$ to represent the out-degree $d_i^-$. The diagonal out-degree matrix of digraph $\calG$ is defined as ${\DD} = {\rm diag}(d_1, d_2, \ldots, d_n)$, and the Laplacian matrix of digraph $\calG$ is defined to be ${\LL}={\DD}-{\AA}$. For a given node $i$, let $N^+_i$ and $N^-_i$ represent its sets of out-neighbors and in-neighbors, respectively.  In other words, $N^+_i =\{ j: (i,j)\in E\}$, and $N^-_i =\{ j: (j,i)\in E\}$. Let $\II$ be the $n$-dimensional identity matrix, and $\ee_i$ be the  $i$-th standard basis column vector, with $i$-th element being 1 and other elements being 0.

The structure of digraph $\calG=(V,E)$ is captured by its adjacency matrix $\AA=(a_{ij})_{n\times n}$, where $a_{ij} = 1$ if there is a directed edge from node $v_i$ to node $v_j$   and $a_{ij} = 0$ otherwise. The in-degree \(d^-_i\) and out-degree \(d^+_i\) of any node \(i\) are defined by \(d^-_i=\sum_{j=1}^n a_{ji}\) and \(d^+_i=\sum_{j=1}^n a_{ij}\), respectively. In the sequel, we use $d_i$ to represent the out-degree $d_i^+$. The  diagonal  degree matrix representing the out-degrees of digraph $\calG$ is ${\DD} = {\rm diag}(d_1, d_2, \ldots, d_n)$, and the Laplacian matrix is ${\LL}={\DD}-{\AA}$. For any given node $i$, $N^+_i$ and $N^-_i$ denote the  sets  of its out-neighbors and in-neighbors, meaning $N^+_i =\{ j: (i,j)\in E\}$ and $N^-_i =\{ j: (j,i)\in E\}$, respectively. Let $\II$ be the $n$-dimensional identity matrix, and $\ee_i$ be the  $i$-th standard basis column vector, with $i$-th element being 1 and other elements being 0.
 
A path $P$ from node $v_1$ to $v_j$ is a sequence of alternating nodes and arcs $v_1$,$(v_1,v_2)$,$v_2$,$\cdots$, $v_{j-1},(v_{j-1},v_j)$, $v_j$ where each node is unique and every arc connects $v_i$ directly to $v_{i+1}$. A loop  is a path plus an arc from the ending node to the starting node. A digraph is   (strongly) connected if there exists a path from one node $v_x$ to another node $v_y$, and vice versa. A digraph is called weakly connected if it is connected when one replaces any directed edge $(i,j)$ with two directed edges $(i,j)$ and $(j,i)$ in opposite directions. A tree is a weakly connected graph with no loops, and an isolated node is considered as a tree. A forest is a particular graph that is a disjoint union of trees.

% In a digraph $\calG$, if for any arc $(i,j)$, the arc $(j,i)$ exists, $\calG$ is reduced to an undirected graph. When $\calG $ is undirected, $a_{ij}= a_{ji}$ holds for an arbitrary pair of nodes $i$ and $j$, and thus $d^+_i= d^-_i$ holds for any  node $i\in V$. Moreover, in undirected graph $\calG $ both adjacency matrix $\AA$ and Laplacian matrix $\LL$ of  $\calG $ are symmetric, satisfying  $\LL\mathbf{1}=\mathbf{0}$.
\subsection{Spanning Converging Forests and Forest Matrix}

In a directed graph $\calG=(V,E)$, a spanning subgraph contains all the nodes from \(V\) and  a subset of edges from \(E\). A rooted converging tree is a weakly connected digraph, where one node, called the root node, has an out-degree of 0, and all other nodes have an out-degree of 1.  An isolated node is considered as a converging tree with the root being itself.   A spanning converging forest of digraph $\calG$ that includes all nodes in \(V\)  and whose weakly connected components are rooted converging trees. Such a forest aligns with the concept of  in-forest as described in~\cite{AgCh01,ChAg02}.

The forest matrix~\cite{ChSh97,ChSh98} is defined as $\mathbf{\Omega}=\left(\II+\LL\right)^{-1}=(\omega_{ij})_{n \times n}$. In the context of digraphs, the forest matrix $\Omega$ is row stochastic, with all its components in the interval $[0,1]$. Moreover,  for each column, the diagonal elements surpass the other elements, that is  $ 0\leq\omega_{ji}< \omega_{ii}\leq 1$ for any pair of different nodes $ i$ and $j $, and the diagonal element $\omega_{ii}$ of matrix $\mathbf{\Omega}$ satisfies $\frac{1}{1+d_i}\leq \omega_{ii} \leq \frac{2}{2+d_i}$~\cite{SuZh23}.

For an unweighted digraph $\calG=(V,E)$, let $\calF $ denote the set of all  its spanning converging forests. For a given spanning converging forest $\phi\in\calF$,  define the root set $\mathcal{R}(\phi )$ of $\phi$ as the collection of roots from all converging trees that constitute $ \phi$, that is, $\mathcal{R}(\phi ) = \{i:(i,j) \notin \phi, \forall j \in V \}$. Since each node $i$ in $\phi$  is part of a specific converging tree, we define a function $r_{\phi}(i): V \rightarrow \calR(\phi) $ mapping node $i$ to the root of its associated converging tree. Thus, if $r_{\phi}(i) = j$, it implies that $j$ is in $\calR(\phi)$, and both nodes $i$ and $j$ are part of the same converging tree  in \(\phi\). Define $ \calF_{ij} $ as the set of spanning converging forests in which nodes $i$ and $j$ are within the same converging tree, rooted at node $j$. Formally, $\calF_{ij} = \{\phi: r_{\phi}(i) = j, \phi \in \calF\}$.  It follows that $\calF_{ii} = \{\phi: i\in \calR(\phi), \phi \in \calF\}$.   For two nodes $i$ and $j$ and a  spanning converging forest $\phi$, define $\mathbb{I}_{\{r_{\phi}(i)=j\}}$ as an indicator function, which is 1 if the input statement is true and 0 otherwise. For example, if $r_{\phi}(i)=j$, $\mathbb{I}_{\{r_{\phi}(i)=j\}}=1$, and $ \mathbb{I}_{\{r_{\phi}(i)=j\}}=0$ otherwise.

\section{Fast Query of Entries in Forest Matrix for Digraphs}
In this section, we propose sampling-based algorithms and novel variance reduction techniques, designed to enable fast querying of entries for the forest matrix.
 
\subsection{ Extension of Wilson's Algorithm }
As mentioned above, the entries of the forest matrix are closely related to the spanning converging forests. In this subsection, we briefly introduce the method for uniformly generating spanning converging forest in graphs, utilizing  an extension of Wilson's algorithm.  

Wilson proposed an algorithm  based on a loop-erased random walk to get a  spanning tree rooted at a given node~\cite{Wi96}.  The loop-erasure technique, pivotal to this algorithm,  is a process derived from the random walk  by performing an erasure operation on its loops in chronological order~\cite{LaFr79,La80}. For a digraph $\calG= (V,E)$,  we can also apply an extension of Wilson's algorithm to get a spanning converging forest $\phi \in \calF$, by using the method similar to that in~\cite{AvLuGaAl18,PiAmBaTr21,SuZh23}, which includes the following  three steps: (i) Construct an augmented digraph $\calG'$  of $\calG$, which is obtained from $\calG$ by adding one new node $\Delta$ to $\calG$, and adding a new edge $(i,\Delta)$  for each node $i$ in $\calG$. (ii) Apply Wilson's algorithm to $\calG'$, designating $\Delta$ as the root, to produce a rooted spanning tree. (iii) Remove node $\Delta$ and its connected edges, and define the root set $\calR$ as the nodes with an out-degree of 0, thereby obtaining a spanning converging forest in $\calG$.

Since Wilson's algorithm returns a uniform rooted spanning tree~\cite{Wi96}, the spanning converging forest obtained using the above steps is also uniformly selected from $\calF$. According to~\cite{SuZh23}, the expected time complexity for generating a uniform spanning converging forest in a digraph $\calG = (V,E)$ is $O(n)$, making this method efficient and practical for large-scale graphs.

\subsection{Probabilistic Interpretation and Unbiased Estimators of Entries in Forest Matrix}
In this subsection, we present a probabilistic interpretation of the entries in the forest matrix and propose unbiased estimators for these entries.

Using the approach in~\cite{Ch82,ChSh06,ChSh97},  for any pair of nodes $i,j\in V$, the entry $\omega_{ij}$ of the forest matrix $\mathbf{\Omega}$ can be expressed as $\omega_{ij}= |\calF_{ij}|/|\calF|$.  This suggests a probabilistic interpretation of the entry $\omega_{ij}$ of the  forest matrix which represents the probability of the root of node $i$ being node $j$   in a uniformly sampled spanning converging forest $\phi \in \calF$. For a spanning converging forest $\phi\in \calF$,  We define an  estimator $\widehat{\omega}_{ij}(\phi)$ for $\omega_{ij}$ as $\widehat{\omega}_{ij}(\phi) = \mathbb{I}_{\{r_{\phi}(i)=j\}}$.   The estimator $\widehat{\omega}_{ij}(\phi)$ takes the value $1$ if the root of $i$ is $j$, and $0$ otherwise. This estimator is unbiased if $\phi$ is uniformly selected from $\calF$, satisfying $\mathbb{E}(\widehat{\omega}_{ij}(\phi)) = \mathbb{P}(r_{\phi}(i)=j) = \frac{|\calF_{ij}|}{|\calF|} = \omega_{ij}$. Moreover, the variance of $\widehat{\omega}_{ij}$ is ${\rm Var}(\widehat{\omega}_{ij}) = \omega_{ij} - \omega_{ij}^2$.

\begin{algorithm}[htbp!]
	\caption{$\textsc{SFQ}(\calG,\calF_0,l,i,j )$}
	\label{alg-rf}
	\Input{  A digraph $\calG$,  a list of $l$ uniformly sampled spanning converging forest    $\calF_0$, a pair of querying id $i,j$	}
	\Output{$\widehat{\omega}_{ij}$ : an estimator of $\omega_{ij}$}
 \textbf{Initialize} :
	$\widehat{\omega}_{ij} \leftarrow 0 $ \\
 \For{$\phi$ in $\calF_0$}{
 $k \leftarrow r_{\phi}(i)$\\
                 \If{$k = j$}{
                $\widehat{\omega}_{ij}\leftarrow$ $\widehat{\omega}_{ij} + 1/l$\\ 
                }
 
}
\textbf{return} $\widehat{\omega}_{ij}$ \;
\end{algorithm}

As previously established, we employ the extension of Wilson's algorithm to generate $l$ spanning converging forests $\phi_{1}, \cdots, \phi_l$. Then we can approximate $\omega_{ij}$ using the mean value $\frac{1}{l}\sum_{k=1}^l\widehat{\omega}_{ij}(\phi_k)$. We detailed this in  Algorithm~\ref{alg-rf}, named as Spanning Forest Query ($\textsc{SFQ}$). The time complexity of Algorithm~\ref{alg-rf} is $O(l)$, where $l$ is the number of the pre-sampled spanning converging forests.

\subsection{Enhanced Estimators with Reduced Variance}
In the preceding subsection, we defined an unbiased estimator $\widehat{\omega}_{ij}(\phi) $ for $\omega_{ij}$.  In this subsection, we introduce enhanced estimators that not only maintain    the unbiased property, but also achieve reduced variance, thus providing more accurate estimations.

We begin with the cases  $i \neq j$. For a spanning converging forest $\phi \in \calF$, the initial estimator $\widehat{\omega}_{ij}(\phi)$ assigns a value of 1 if a path exists from $i$ to $j$ in $\phi$, and 0 otherwise. However, this approach overlooks the case where a node $k \in N^-_j$ (indicating an edge $(k,j) \in E$) is the root for node $i$ in the forest $\phi$. This situation suggests that node $i$ might be rooted at node $j$ in other forests due to the possible existence of a path  from $i$ to $j$, which is a concatenation of a path from $i$ to $k$ and an edge $(k,j)$. Our new approach aims to incorporate these overlooked instances by aggregating information from forests where node $i$ is rooted in the in-neighbors of node $j$.   

Specifically, for nodes $k \in N^-(j)$,  $\omega_{ik}$  denotes the probability of node $i$  rooted at $k$. When performing the loop-erased random walk in the extension of Wilson's algorithm, if the walk is currently at node $i$, it has a probability of $\frac{1}{1+d_i}$   moving to a random out-neighbor or becoming a root (moving to the extended node $\Delta$ in the augmented graph $\calG'$). Let's consider a loop-erased random walk that starts at node $i$ and ends at node $k$, creating a path $P$ from $i$ to $k$ and then terminating.  Now, imagine a walking path $P'$  slightly different from path $P$, where  the walker does not  ending at $k$, but continues to jump to node $j$ and then terminates.  The probability of  path $P'$  occurring is \(\frac{1}{1+d_j}\) times that of $P$. Since any path  from $i$ to $j$ must pass through a node $k \in N^-(j)$ before reaching $j$, we define a new estimator $\widetilde{\omega}_{ij}(\phi)$ for different nodes \(i, j\), and spanning converging forest \(\phi \in \calF\) as $\widetilde{\omega}_{ij}(\phi) = \frac{1}{1+d_j}\sum_{k \in N^-(j)}\widehat{\omega}_{ik}(\phi).$

The following lemma  shows that $\widetilde{\omega}_{ij}$ is an unbiased estimator of $\omega_{ij}$ and has a reduced variance compared to $\widehat{\omega}_{ij}$:

\begin{lemma}\label{le-tildeomega}
       For two different nodes $i$ and $j$ in graph $\calG$  and a uniformly chosen spanning converging forest $\phi \in \calF$, $\widetilde{\omega}_{ij}(\phi) = \frac{1}{1+d_j}\sum_{k \in N^-(j)}\widehat{\omega}_{ik}(\phi)$ is an unbiased estimator of $\omega_{ij}$.   The variance of this estimator is given by ${\rm Var}(\widetilde{\omega}_{ij}) = \frac{\omega_{ij}}{1+d_j} - \omega_{ij}^2$, which is always less than or equal to the variance of the estimator $\widehat{\omega}_{ij}$. 
\end{lemma}

\begin{proof}
   Using the relationship  $\mathbf{\Omega}(\II+\LL) = \II$, it follows that $\ee_i^\top\mathbf{\Omega}(\II+\LL)\ee_j = 0$. This implies $(1+d_j)\omega_{ij} -\sum_{k\in N^-_j}\omega_{ik} = 0$, leading to $ {\omega}_{ij}  = \frac{1}{1+d_j}\sum_{k \in N^-_j} {\omega}_{ik}. $ Since $\mathbb{E}(\widehat{\omega}_{ik}) ={\omega}_{ik} $, we have $ \mathbb{E}(\widetilde{\omega}_{ij}) =\frac{1}{1+d_j}\sum_{k \in N^-(j)}\mathbb{E}(\widehat{\omega}_{ik}) = \frac{1}{1+d_j} \sum_{k \in N^-_j} {\omega}_{ik} = {\omega}_{ij}, $  which shows that $\widetilde{\omega}_{ij}$ is an unbiased estimator of $\omega_{ij}$.

    Next, we calculate the variance of $\widetilde{\omega}_{ij}$. Considering   $\mathbb{E}(\widehat{\omega}_{ik}^2) = {\omega}_{ik}$ and $\mathbb{E}(\widehat{\omega}_{ij}\widehat{\omega}_{ik}) =0$  for $j\neq k$, we have $ {\rm Var}(\widetilde{\omega}_{ij})  =\mathbb{E}(\widetilde{\omega}_{ij}^2) -  (\mathbb{E}(\widetilde{\omega}_{ij}))^2 =\mathbb{E}(\frac{1}{(1+d_j)^2} (\sum_{k \in N^-(j)}\widehat{\omega}_{ik})^2) -  {\omega}_{ij}^2= \frac{\sum_{k \in N^-(j)}\widehat{\omega}_{ik}}{(1+d_j)^2}  -   {\omega}_{ij}^2  = \frac{\omega_{ij}}{1+d_j} - \omega_{ij}^2$,  which finishes the proof. 
\end{proof}

Lemma~\ref{le-tildeomega} demonstrates that $\widetilde{\omega}_{ij}$ is an unbiased estimator with a lower variance compared to $\widehat{\omega}_{ij}$. The reduction in variance is attributed to $\widetilde{\omega}_{ij}$ accounting for instances where node $i$ is rooted at the in-neighbors of node $j$, a scenario that $\widehat{\omega}_{ij}$ fails to consider. Specifically, if $\phi \in \calF_{ik}$ with $k \in N^-_j$, then $\widetilde{\omega}_{ij}(\phi)$ equals $\frac{1}{1+d_i}$, while $\widehat{\omega}_{ij}(\phi)$ is 0. Conversely, for $\phi \in \calF_{ij}$, $\widehat{\omega}_{ij}(\phi)$ is 1, but $\widetilde{\omega}_{ij}(\phi)$ is 0. This indicates that $\widetilde{\omega}_{ij}$ partially disregards instances where $i$ is directly rooted at $j$.

To address this issue, we introduce an estimator $\widebar{\omega}_{ij}(\phi,\alpha)$, which is a linear combination of $\widetilde{\omega}_{ij}$ and $\widehat{\omega}_{ij}$, defined as $\widebar{\omega}_{ij}(\phi,\alpha) = \alpha \widehat{\omega}_{ij}(\phi) + (1-\alpha)\widetilde{\omega}_{ij}(\phi)$. The parameter $\alpha$ is chosen to minimize the variance of the estimator, as explained in the following lemma:

\begin{lemma}\label{le-baromega}
    For two distinct nodes   $i$ and $j$ in $V$, a parameter $\alpha \in \calR$, and a uniformly chosen spanning converging forest $\phi \in \calF$, the estimator $\widebar{\omega}_{ij}(\phi,\alpha) = \alpha \widehat{\omega}_{ij}(\phi) + (1-\alpha)\widetilde{\omega}_{ij}(\phi)$ is an unbiased estimator for  $\omega_{ij}$. The variance of  estimator $\widebar{\omega}_{ij}(\phi,\alpha)$ is minimized when $\alpha = \frac{1}{2+d_j}$. Defining $\widebar{\omega}_{ij}(\phi)$ as $\widebar{\omega}_{ij}(\phi, \frac{1}{2+d_j})=  \frac{1}{2+d_j}(\widehat{\omega}_{ij}(\phi)+\sum_{k\in N^-_j}\widehat{\omega}_{ik}(\phi))$, the variance of $\widebar{\omega}_{ij}(\phi)$ is given by ${\rm Var}(\widetilde{\omega}_{ij}) = \frac{\omega_{ij}}{2+d_j} - \omega_{ij}^2$, which is smaller than that of both $\widehat{\omega}_{ij}(\phi)$ and $\widetilde{\omega}_{ij}(\phi)$.  Moreover, the non-negativity of variance implies $\omega_{ij}\leq \frac{1}{2+d_j}$.
\end{lemma}

\begin{proof}
   Since both $\widehat{\omega}_{ij}$ and $\widetilde{\omega}_{ij}$ are unbiased estimators for $\omega_{ij}$, it follows that  $\mathbb{E}(\widebar{\omega}_{ij}(\phi,\alpha)) = \alpha\mathbb{E}( \widehat{\omega}_{ij}(\phi)) + (1-\alpha)\mathbb{E}(\widetilde{\omega}_{ij}(\phi)) = \omega_{ij}$, indicating the estimator $\widebar{\omega}_{ij}(\phi,\alpha)$ is also unbiased for $\omega_{ij}$. We proceed to calculate the variance of $\widebar{\omega}_{ij}(\phi,\alpha)$ as  \begin{equation*}
        \begin{aligned}
        &\quad {\rm Var}(\widebar{\omega}_{ij}(\phi,\alpha))  =\mathbb{E}(\widebar{\omega}_{ij}^2(\phi,\alpha)) -  (\mathbb{E}(\widebar{\omega}_{ij}(\phi,\alpha))^2\\ &= \frac{\omega_{ij}(2+d_j)}{1+d_j}\left(\alpha - \frac{1}{2+d_j}\right)^2 + \frac{\omega_{ij}}{2+d_j} - \omega_{ij}^2.
        \end{aligned}
    \end{equation*}

    The variance of the  estimator $\widebar{\omega}_{ij}(\phi,\alpha)$ is minimized when $\alpha = \frac{1}{2+d_j}$, resulting in $\frac{\omega_{ij}}{2+d_j} - \omega_{ij}^2$, which finishes the proof.
\end{proof}

Lemma~\ref{le-baromega} indicates that the newly formulated estimator $\widebar{\omega}_{ij}(\phi) = \frac{1}{2+d_j}(\widehat{\omega}_{ij}(\phi)+\sum_{k\in N^-_j}\widehat{\omega}_{ik}(\phi))$ has a lower variance than the previous two estimators by incorporating their respective information.  Having established the improved estimator $\widebar{\omega}_{ij}$ for the scenario   $i \neq j$, we next  focus on the situation that $i$ is equal to $j$.

For the case $i=j$, we consider a similar approach of aggregating information from  the in-neighbors of node $i$. This leads to a new estimator with reduced variance compared to $\widehat{\omega}_{ii}$.     We define the new estimator $\widebar{\omega}_{ii}(\phi)$ as  $\widebar{\omega}_{ii}(\phi) = \frac{1}{1+d_i}(1+\sum_{k\in N^-_i}\widehat{\omega}_{ik}(\phi))$. Lemma~\ref{le-baromgaii} demonstrates that $\widebar{\omega}_{ii}(\phi)$ is an unbiased estimator for $\omega_{ii}$ with reduced variance.
 
\begin{lemma}\label{le-baromgaii}
   For node $i\in V$ and a uniformly chosen spanning converging forest $\phi \in \calF $, $\widebar{\omega}_{ii}(\phi) = \frac{1}{1+d_i}(1+\sum_{k\in N^-_i}\widehat{\omega}_{ik}(\phi))$ is an unbiased estimator for $\omega_{ii}$. The variance of this estimator is  ${\rm Var}(\widebar{\omega}_{ii}) = \frac{3\omega_{ii}}{1+d_i} - \frac{2}{(1+d_i)^2} -\omega_{ii}^2$, which is always less than or equal to the variance of the estimator $\widehat{\omega}_{ii}$. 
\end{lemma}

\begin{proof}
Since  $\mathbf{\Omega}(\II+\LL) = \II$, it follows that $\ee_i^\top\mathbf{\Omega}(\II+\LL)\ee_i = 1$. This implies   $(1+d_i)\omega_{ii} -\sum_{k\in N^-_i}\omega_{ik} = 1$, that is,  $ {\omega}_{ii}  = \frac{1}{1+d_i}(1 + \sum_{k \in N^-_i} {\omega}_{ik}). $ Since $\mathbb{E}(\widehat{\omega}_{ik}) ={\omega}_{ik} $, we have $ \mathbb{E}(\widebar{\omega}_{ii}) =\frac{1}{1+d_i}(1+\sum_{k \in N^-(i)}\mathbb{E}(\widehat{\omega}_{ik})) = \frac{1}{1+d_i}(1 + \sum_{k \in N^-_i} {\omega}_{ik})= {\omega}_{ii}, $  which shows that $\widetilde{\omega}_{ij}$ is an unbiased estimator for $\omega_{ij}$. The variance of $\widebar{\omega}_{ii}$ can be derived as follows: ${\rm Var}(\widebar{\omega}_{ii})= \mathbb{E}(\widebar{\omega}_{ii})^2 - (\mathbb{E}(\widebar{\omega}_{ii}))^2 = \frac{1}{(1+d_i)^2}\mathbb{E}(( 1+\sum_{k\in N^-_i}\widehat{\omega}_{ik})^2 ) - \omega_{ii}^2 
     = \frac{1}{(1+d_i)^2}\mathbb{E}(1+2\sum_{k\in N^-_i}\widehat{\omega}_{ik} + (\sum_{k\in N^-_i}\widehat{\omega}_{ik})^2 ) - \omega_{ii}^2 
     = \frac{1+3\sum_{k\in N^-_i}{\omega}_{ik}}{(1+d_i)^2}-\omega_{ii}^2  =  \frac{1+3((1+d_i)\omega_{ii}-1)}{(1+d_i)^2} - \omega_{ii}^2 =\frac{3\omega_{ii}}{1+d_i} - \frac{2}{(1+d_i)^2}  -\omega_{ii}^2.$ Then we get the following relation $ {\rm Var}\{\widehat{\omega}_{ii}\} - {\rm Var}\{\widebar{\omega}_{ii}\}    =\frac{2(1-\omega_{ii})}{(1+d_i)^2}+ \frac{d_{i}(d_i-1)\omega_{ii}}{(1+d_i)^2}\geq 0.$ This inequality  shows that the variance of $\widebar{\omega}_{ii}$  is no more than  the variance of the estimator $\widehat{\omega}_{ii}$, which completes the proof.
\end{proof}

In the above,  we have proposed enhanced estimators   $\widebar{\omega}_{ii}$ and $\widebar{\omega}_{ij}$ for the diagonal  $\omega_{ii}$ and non-diagonal entries $\omega_{ij}$. Suppose that we already have used the extension of Wilson's algorithm to generate $l$ spanning converging forests $\phi_{1}, \cdots, \phi_l$. Then we can approximate $ \omega_{ii}$   and $\omega_{ij}$ using  $\frac{1}{l}\sum_{k=1}^l\widebar{\omega}_{ii}(\phi_k)$ and $\frac{1}{l}\sum_{k=1}^l\widebar{\omega}_{ij}(\phi_k)$. We detailed this in  Algorithm~\ref{alg-rf+}, named as Spanning Forest Query Plus ($\textsc{SFQPlus}$).

\begin{algorithm}[htbp!]
	\caption{$\textsc{SFQPlus}(\calG,\calF_0,l,i,j )$}
	\label{alg-rf+}
	\Input{  A digraph $\calG = (V,E)$,  a list of $l$ uniformly sampled spanning converging forest    $\calF_0$, a pair of querying id $i,j$	}
	\Output{$\widebar{\omega}_{ij}$ : an estimator of $\omega_{ij}$}
\If{$i = j$}{
 \textbf{Initialize} :
	$\widebar{\omega}_{ij} \leftarrow \frac{1}{1+d_i} $ \\

}\Else{ \textbf{Initialize} :
	$\widebar{\omega}_{ij} \leftarrow  0 $ \\
}

 \For{$\phi$ in $\calF_0$}{
 $k \leftarrow r_{\phi}(i)$\\
\If{$k = j$ \& $i\neq j$ }{
    $\widebar{\omega}_{ij}\leftarrow$ $\widebar{\omega}_{ij} + \frac{1}{(2+d_j)l}$

}\ElseIf{ edge $(k,j) \in E$}{
\If{$i = j$}{

 $\widebar{\omega}_{ij}\leftarrow$ $\widebar{\omega}_{ij} + \frac{1}{(1+d_j)l}$

}\Else{$\widebar{\omega}_{ij}\leftarrow$ $\widebar{\omega}_{ij} + \frac{1}{(2+d_j)l}$}

}

}
\textbf{return} $\widebar{\omega}_{ij}$ \;
\end{algorithm}

The time complexity of Algorithm~\ref{alg-rf+} is $O(l)$, where $l$ is the number of the pre-sampled spanning converging forests.  As we increase the number of pre-sampled  forest $l $, we observe a corresponding decrease in the estimation error between $\widebar{\omega}_{ij}$ and the actual value $\omega_{ij}$. To quantify this relationship, we introduce Theorem~\ref{th-rf}, which specifies the necessary size of $l$ to achieve a necessary  error guarantee with a high probability. Before giving Theorem~\ref{th-rf}, we introduce the following lemma. 
 
\begin{lemma}(Chernoff bound~\cite{ChLu06})\label{le-chernoff}
	Let $ x_i(1\leq i\leq l) $ be independent random variables satisfying $ |x_i- \mathbb{E}\{x_i\}|\leq M$ for all $ 1\leq i \leq l $. Let $ x = \frac{1}{l} \sum_{i=1}^l x_i $. Then we have
	\begin{equation}\label{key}
		\mathbb{P}(|x-\mathbb{E}\{x\}| \leq \epsilon )\geq 1-2\exp{\left(-\frac{l\epsilon^2}{2({\rm Var}\{x\}l+M\epsilon/3)}\right)}.
	\end{equation}
\end{lemma}

\begin{theorem}\label{th-rf}
     For any pair of  nodes $i\neq j$, and  parameters  $ \epsilon,\sigma,\delta \in(0,1) $, if $l$ is chosen obeying   $ l =\left \lceil \frac{1}{(2+d_j)^2}     ( \frac{1}{2\epsilon^2} + \frac{2}{3\epsilon }    )         \log(\frac{2}{\delta} ) \right \rceil $, then the approximation $\widebar{\omega}_{ij}$ of $ \omega_{ij}$  returned by  Algorithm~\ref{alg-rf+} satisfies the following relation with  probability of at least $1-\delta$:
	\begin{equation}\label{eq-omegaij}
		  \omega_{ij}-\epsilon 	\leq 	\widebar{\omega}_{ij}   \leq   \omega_{ij} +\epsilon.
	\end{equation} 
        For the case that $i = j$, and for parameters  $ \epsilon,\delta \in(0,1) $, if $l$ is chosen obeying    $l =\left \lceil  (\frac{2}{3\epsilon}+\frac{1}{4\epsilon^2})\log(\frac{2}{\delta})  \right \rceil$, then the approximation $\widebar{\omega}_{ii}$ of $ \omega_{ii}$  returned by  Algorithm~\ref{alg-rf+} satisfies the following relation with    probability of at least $1-\delta$:
	\begin{equation}\label{eq-omegaii}
		(1-\epsilon) \omega_{ii}	\leq 	\widebar{\omega}_{ii}   \leq (1+\epsilon) \omega_{ii}.
	\end{equation} 
\end{theorem}
    
\begin{proof}
 For the case that $i\neq j$, the output $\widebar{\omega}_{ij}$ of Algorithm~\ref{alg-rf+} is  $\widebar{\omega}_{ij} = \frac{1}{l} \sum_{k=1}^l \widebar{\omega}_{ij}(\phi_k)$. Since the variance of $\widebar{\omega}_{ij}(\phi_k)$ is $\frac{\omega_{ij}}{2+d_j} - \omega_{ij}^2$ for $k = 1,\cdots, l$, and the $l$ variables are independent,  we can compute the variance of $\widebar{\omega}_{ij}$ as ${\rm Var}(\widebar{\omega}_{ij}) =\frac{1}{l}(\frac{\omega_{ij}}{2+d_j} - \omega_{ij}^2)$. To obtain the  absolute error bound given by~\eqref{eq-omegaij},  $\widebar{\omega}_{ij}$ needs to satisfy 
 $$\mathbb{P}(| \widebar{\omega}_{ij}-\omega_{ij}| \geq \epsilon   )\leq \delta.$$

We now show that the above relation holds true. To this end, we designate $x_j = \widebar{\omega}_{ij}(\phi_k)$ for $1 \leq k \leq l$ and $x = \widebar{\omega}_{ij}$, and invoke the Chernoff bound in Lemma~\ref{le-chernoff}. Then, in order to prove~\eqref{eq-omegaij}, we only need to show that the following inequality holds:
$$2\exp{\left(-\frac{l\epsilon^2 }{2({\rm Var}\{\widebar{\omega}_{ij}\}l+M\epsilon\omega_{ij}/3)}\right)} \leq \delta,$$
which leads to
\begin{equation}\label{ltomeet}
l \geq \log(\frac{2}{\delta})\left(\frac{2{\rm Var}\{\widebar{\omega}_{ij}\}l}{\epsilon^2 }+\frac{2M\omega_{ij}}{3\epsilon}\right).
\end{equation}
Since $|\widehat{\omega}_{ij} - \omega_{ij}| \leq \frac{1}{2+d_j}$, we can set $M=\frac{1}{2+d_j}$. Considering  that ${\rm Var}\{ \widebar{\omega}_{ij}(\phi_j)\} = \frac{1}{l}(\frac{\omega_{ij}}{2+d_j} - \omega_{ij}^2)$,   inequality~\eqref{ltomeet} is reduced to:

\begin{equation}
    l \geq \log(\frac{2}{\delta})\left(\frac{2\omega_{ij}}{\epsilon^2(2+d_j)}
+\frac{2\omega_{ij}}{3\epsilon(2+d_j)}-\frac{2\omega_{ij}^2}{\epsilon^2}\right).
\end{equation}

Thus, for any pair of  nodes $i\neq j$, since $0\leq \omega_{ij}\leq \frac{1}{2+d_j}$,  selecting 
$ l =\left \lceil    \left( \frac{1}{2\epsilon^2} + \frac{2}{3\epsilon }    \right) \frac{\log{\frac{2}{\delta} }}{(2+d_j)^2}  \right \rceil $ ensures the required inequality~\eqref{ltomeet} always holds. This completes the proof.

For the case that $i = j$, the output $\widebar{\omega}_{ii}$ of Algorithm~\ref{alg-rf+} is  $\widebar{\omega}_{ii} = \frac{1}{l} \sum_{k=1}^l \widebar{\omega}_{ii}(\phi_k)$.   For a spanning converging forest $\phi \in \calF$,  $\widebar{\omega}_{ii}(\phi)$ is either $\frac{1}{1+d_i}$ or $\frac{2}{1+d_i}$. Considering   $\frac{1}{1+d_i}\leq \omega_{ii} \leq \frac{2}{2+d_i}$, it follows that $|\widebar{\omega}_{ii} - \omega_{ii}| \leq \frac{1}{1+d_i}$. Then we can set the bound $M$ as $M=\frac{1}{1+d_i}$.  Employing a similar analytical approach as above,  the number of $l$ needs to satisfy the following inequality:
  \begin{equation}\label{l}
      l \geq \log(\frac{2}{\delta})\left(\frac{2{\rm Var}\{\widebar{\omega}_{ii}\}l}{\epsilon^2\omega_{ii}^2}+\frac{2}{3(1+d_i)\epsilon\omega_{ii}}\right). 
 \end{equation}
Given that $\widebar{\omega}_{ii} = \frac{1}{l}(\frac{3\omega_{ii}}{1+d_i}-\frac{2}{(1+d_i)^2}-\omega_{ii}^2)$, we derive that
    \begin{equation}
    \begin{aligned}
              & \frac{{\rm Var}(\widebar{\omega}_{ii})l}{\omega_{ii}^2}  = \frac{3}{(1+d_i)\omega_{ii}} - \frac{2}{(1+d_i)^2\omega_{ii}^2} -1 \\&= -\frac{2}{(1+d_i)^2}(\frac{1}{\omega_{ii}}-\frac{3(1+d_i)}{4})^2+\frac{1}{8}  \leq \frac{1}{8}
    \end{aligned}
    \end{equation} 
Thus, with $\frac{{\rm Var}(\widebar{\omega}_{ii})l}{\omega_{ii}^2} \leq \frac{1}{8}$ and $(1+d_i)\omega_{ii} \leq 1$, choosing $l = \left \lceil (\frac{2}{3\epsilon} + \frac{1}{4\epsilon^2})\log\left(\frac{2}{\delta}\right) \right \rceil$ ensures that inequality~\eqref{l} is always satisfied, which completes the proof.
\end{proof}

Based on Theorem~\ref{th-rf}, when the parameters $\epsilon$, $\sigma$, and $\delta$ are fixed constants, the number of spanning converging forests $l$ required by Algorithm~\ref{alg-rf+} to achieve an $\epsilon$ absolute error for non-diagonal entries and relative error for diagonal entries does not increase with the expansion of the graph. Therefore, the time complexity of Algorithm~\ref{alg-rf+} can be considered as $O(1)$ for achieving a satisfactory error guarantee, regardless of the graph size.

\section{Fast Query of Entries of Forest Matrix for Evolving Graphs}
In the previous section, we developed estimators approximating the entries of the forest matrix in static directed graphs. However, real-world graphs often evolve dynamically. This section addresses evolving graphs,  focusing on two types of updates: edge insertion and edge deletion.

\subsection{Problem Statement}
Consider an initial graph $\calG_0 = (V_0,E_0)$ and a corresponding list $L_0$ of $l_0$ spanning converging forests, which is uniformly sampled using an extension of Wilson's algorithm.  In this dynamic context, there are two possible requests: updates  and  queries. For the request of updates, we restrict our focus to edge insertions and deletions. Other updates such as nodes insertions/deletions can be easily converted to a sequence of edges insertion/deletions. For the $k$-th edge update, let $e_k = (u_k,v_k)$ denote the edge being modified, resulting in an updated graph $\calG_k = (V_k,E_k)$. Specifically, for edge insertions, $E_{k} = E_{k-1} \cup \{e_k\}$, and for deletions, $E_{k} = E_{k-1} \setminus \{e_k\}$. 
 
Query requests involve asking for specific values of the forest matrix entries of $\calG_k$, denoted by $\mathbf{\Omega}_k$. Let $\calF(\calG_k)$  denote the set of all spanning converging forests of graph $\calG_k$. Utilizing the methods described earlier, if we have spanning converging forests uniformly sampled from $\calF(\calG_k)$, Algorithm \textsc{SFQPlus} can provide an estimation. However, resampling these forests after every update requires  time complexity of $O(ln)$~\cite{SuZh23}, which is inefficient. This naturally leads to the following question.     Suppose that $L_{k-1} $ is a spanning converging forest list  with $l_{k-1}$ spanning converging forests uniformly sampled from the set $\calF(\calG_{k-1})$.  When an update $e_k = (u_k,v_k)$ occurs, can we develop an efficient method to adapt the list $L_{k-1}$ into $L_k$, ensuring that each spanning converging forest in set  $\calF(\calG_k)$ has an equal probability  appearing in the updated forest list $L_k$? In the following, we address this challenge by proposing  a method with an expected cost of $O(1)$ for updating the forest list, while preserving the uniform  property of the sampling.

\subsection{Edge Insertion}
In this subsection, we delve into the edge insertion update. Consider the $k$-th update involving the insertion of an edge $e_k = (u_k,v_k)$ to the graph $\calG_{k-1}$. We start with a spanning converging forest list $L_{k-1} = \{\phi_1, \cdots, \phi_{l_{k-1}}\}$ with $l_{k-1}$ spanning converging forests. Each forest from the set $\calF(\calG_{k-1})$ is equally likely to be included in  list $L_{k-1}$. Our goal is to modify $L_{k-1}$ into $L_k$ such that the forests in  $L_k$ are uniformly sampled from the updated set $\calF(\calG_{k})$.

 Given that $E_k = E_{k-1} \cup \{e_k\}$, it follows that $\calF(\calG_{k-1})$ is a subset of $\calF(\calG_k)$, indicating that all forests in $\calF(\calG_{k-1})$ are also contained in $\calF(\calG_k)$. We now focus on those spanning converging forests that are in $\calF(\calG_k)$ but not in $\calF(\calG_{k-1})$. Define $\Delta \calF_k$ as $\calF(\calG_k) \setminus \calF(\calG_{k-1})$. It's clear that $\Delta \calF_k$ is non-empty, and all its forests include the newly added edge $e_k = (u_k,v_k)$.  We define a subset $\calF(\calG'_{k-1}) \subset \calF(\calG_{k-1})$, where $\calF(\calG'_{k-1}) = \{\phi: \phi \in \calF(\calG_{k-1}), r_{\phi}(u_k) = u_k, r_{\phi}(v_k) \neq u_k\}$. Lemma~\ref{le-edge}  establishes a bijection between $\calF(\calG'_{k-1})$ and $\Delta \calF_k$.

\begin{lemma}\label{le-edge}
For any spanning converging forest $\phi \in \Delta \calF_k$, the forest $\phi' = \phi \setminus \{e_k\}$ belongs to $\calF(\calG'_{k-1})$. This mapping  constitutes a bijection between $\calF(\calG'_{k-1})$ and $\Delta \calF_k$.
\end{lemma}

\begin{proof}
Consider a spanning converging forest $\phi \in \Delta \calF_k$, which includes the newly added edge $e_k = (u_k,v_k)$. By removing this edge, we obtain the forest $\phi' = \phi \setminus \{e_k\}$. Since all edges of $\phi'$ are part of the graph $\calG_{k-1}$, it is evident that $\phi'$ belongs to $\calF(\calG_{k-1})$. Moreover, given that $r_{\phi'}(u_k) = u_k$ and $r_{\phi'}(v_k) \neq u_k$, $\phi'$ is also a part of $\calF(\calG'_{k-1})$. 

Moreover, for any two distinct spanning converging forests $\phi_1, \phi_2 \in \Delta \calF_k$, their corresponding forests $\phi_1'$ and $\phi_2'$ obtained by deleting edge $e_k$ are also distinct. Additionally, for any spanning forest $\phi' \in \calF(\calG'_{k-1})$, we consider the forest $\phi = \phi' \cup \{e_k\}$. The conditions $r_{\phi'}(u_k) = u_k$ and $r_{\phi'}(v_k) \neq u_k$ guarantee that $\phi$ is well-defined and belongs to $L_k$, which finishes the proof.
\end{proof}

With Lemma~\ref{le-edge}, we propose the edge insertion update Algorithm~\ref{alg-addedge}  to update the forest list $L_{k-1}$ to $L_k$.

\begin{algorithm}[htbp!]
	\caption{$\textsc{Insert-Update}(\calG_{k-1},L_{k-1},l_{k-1},e_k )$}
	\label{alg-addedge}
	\Input{  A digraph $\calG_{k-1} = (V_{k-1},E_{k-1})$,  a list of $l_{k-1}$  spanning converging forest  $L_{k-1}$ uniformly sampled from $\calF(\calG_{k-1})$, an edge $e_k = (u_k,v_k)$	to be inserted}
	\Output{An updated spanning converging forest list $L_k$}
{
 \textbf{Initialize} :
	$L_k \leftarrow L_{k-1} $ \\
}

 \For{$\phi$ in $L_{k-1}$}{

\If{$r_{\phi}(u_k) = u_k$ \& $r_{\phi}(v_k) \neq u_k$ }{
    $\widehat{\phi} \leftarrow  \phi\cup \{e_k\}$\\
    \textbf{Add} $\widehat{\phi}$  to   $ L_k  $
} 
}
\textbf{return} $L_k$ \;
\end{algorithm}

\begin{theorem}\label{th-addedge}
    For a spanning converging forest list $L_{k-1}$ with $l_{k-1}$ forests uniformly sampled from $\calF(\calG_{k-1})$, and the insertion edge $e_k = (u_k,v_k)$,    all forests in the updated set   $\calF(\calG_{k})$ have the same probability to be included  in  the list returned by Algorithm~\ref{alg-addedge}.
\end{theorem}

\begin{proof}
%To prove that distinct forests $\phi_1$ and $ \phi_2$ in $\calF(\calG_k)$ have identical probability  of appearing in $L_k$, 
We partition $\calF(\calG_k)$ into two disjoint subsets: $\calF(\calG_{k-1})$ and $\Delta \calF_k$, with $\Delta \calF_k = \calF(\calG_k) \setminus \calF(\calG_{k-1})$.  In order to prove the theorem, we distinguish four cases: (i)  $\phi_1 \in \calF(\calG_{k-1})$ and $\phi_2 \in \calF(\calG_{k-1})$, (ii)  $\phi_1 \in \Delta \calF_k$ and $\phi_2 \in \Delta \calF_k$, (iii)  $\phi_1 \in \Delta \calF(\calG_{k-1})$ and $\phi_2 \in \Delta \calF_k$, and (iv)  $\phi_1 \in \Delta \calF_k$ and $\phi_2 \in \Delta \calF(\calG_{k-1})$.
Moreover, for the convenience of description, let  $\mathbb{P}(\phi_1 \in L_k)$ denote the probability that  forest $\phi_1$ is in $L_k$. In a similar way, we can define other probabilities.  Note that all forests in  $L_{k-1}$ are uniformly sampled from $\calF(\calG_{k-1})$.
Then, the theorem can be proved as follows.

For the first case that $\phi_1 \in \calF(\calG_{k-1})$ and $\phi_2 \in \calF(\calG_{k-1})$,  we have $\mathbb{P}(\phi_1 \in L_k) = \mathbb{P}(\phi_1 \in L_{k-1}) = \mathbb{P}(\phi_2 \in L_{k-1}) = \mathbb{P}(\phi_2 \in L_k)$.

For the second case that $\phi_1 \in \Delta \calF_k$ and $\phi_2 \in \Delta \calF_k$, define $\phi_1' = \phi_1 \setminus \{e_k\}$ and $\phi_2' = \phi_2 \setminus \{e_k\}$. Then we have that $\mathbb{P}(\phi_1 \in L_k) = \mathbb{P}(\phi_1' \in L_{k-1}) = \mathbb{P}(\phi_2' \in L_{k-1}) = \mathbb{P}(\phi_2 \in L_k)$.

For the third case that $\phi_1 \in \Delta \calF(\calG_{k-1})$ and $\phi_2 \in \Delta \calF_k$,  define $\phi_2' = \phi_2 \setminus \{e_k\}$. Then we have  $\mathbb{P}(\phi_1 \in L_k) = \mathbb{P}(\phi_1 \in L_{k-1}) = \mathbb{P}(\phi_2' \in L_{k-1}) = \mathbb{P}(\phi_2 \in L_k)$. 

For the fourth case that $\phi_1 \in \Delta \calF_k$ and $\phi_2 \in \Delta \calF(\calG_{k-1})$, using the same approach as  the third case, we obtain $\mathbb{P}(\phi_1 \in L_k) = \mathbb{P}(\phi_2 \in L_k)$. 

Thus, for two  distinct forests $\phi_1, \phi_2 \in \calF(\calG_k)$, they have  equal probability to appear in $L_k$. which finishes the proof.
\end{proof}

 Theorem~\ref{th-addedge} demonstrates that the list $L_k$, generated by Algorithm~\ref{alg-addedge}, achieves uniform sampling from the updated set $\calF(\calG_{k})$. Additionally, the time complexity of Algorithm~\ref{alg-addedge} is $O(l_{k-1})$, where $l_{k-1}$ is the number of forests in the initial list $L_{k-1}$. 

 \begin{example}
Consider a simple graph $\calG_0$  in Figure~\ref{ftoyadd}, which has $3$ nodes and $3$ edges. Graph $\calG_0$ comprises 7 spanning converging forests, highlighted in the first row of Figure~\ref{ftoyadd}. After the insertion of edge $e  = (1,3)$, the graph updates to $\calG_1$, which contains 9 spanning converging forests forming $\calF(\calG_1)$, as demonstrated in the second row of Figure~\ref{ftoyadd}. Suppose that we have a forest list $L_0$  uniformly sampled from $\calF(\calG_0)$. A straightforward  interpretation of Algorithm~\ref{alg-addedge}  for updating the list $L_0$ to $L_1$ involves two steps: first creating a copy of $L_0$, and then attempting to add the new edge into each forest  in the list $L_0$. For the 7  forests in $L_0$, only two forests $\phi_1$ and $\phi_2$, distinguished by a yellow background, evolve into the valid forests $\widehat{\phi}_1$ and $\widehat{\phi}_2$ with the incorporation of the edge $e = (1,3)$.  The addition of edge $e$ to other forests either generates a cycle or results in a node with an out-degree of 2, both of which contradict the definition of spanning converging forests.  Lemma~\ref{le-edge} theoretically validates that there is a bijection   between the sets $\{\phi_1,\phi_2 \}$ and $\{\widehat{\phi}_1,\widehat{\phi}_2 \}$, which validates that each forest in $\calF(\calG_1) $ has the same probability to appear in $L_1$. 
 
 \end{example}

\begin{figure}[htbp!]
	\centering
	\includegraphics[width=1\columnwidth]{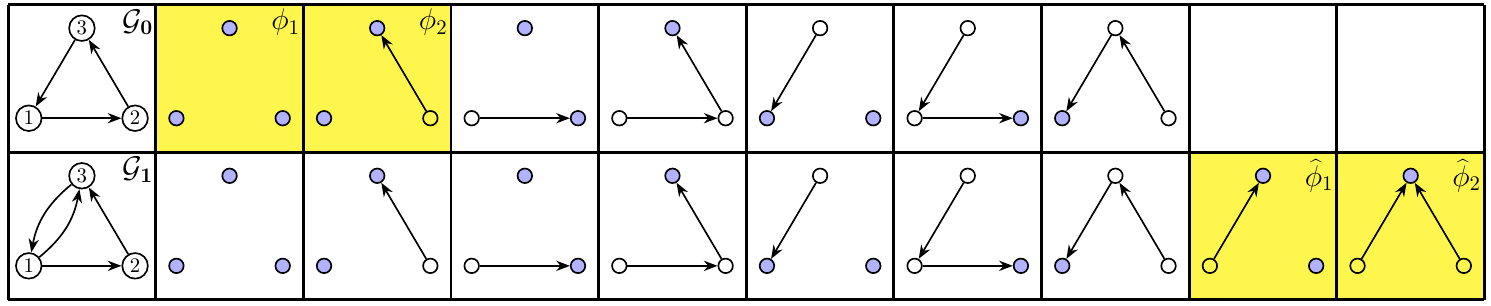 }
	\caption{A toy digraph $\calG_0$ and updated graph $\calG_1$ with their spanning converging forests.  Blue nodes are roots.}\label{ftoyadd}	
\end{figure}

\subsection{Edge Deletion}
In this subsection, we consider the edge deletion update. Specifically, we consider the deletion of an edge $e_k = (u_k,v_k)$ from the graph  $\calG_{k-1}$, resulting in the updated graph $\calG_k = \calG_{k-1} \setminus \{e_k\}$. Our objective is to adapt the initial forest list $L_{k-1}$ into $L_k$, ensuring the uniformity property of the sampling is preserved.

Given that the updated edge set $E_k$ is defined as $E_k = E_{k-1} \setminus \{e_k\}$, it is evident that $\calF(\calG_{k})$ is a subset of $\calF(\calG_{k-1})$. A straightforward approach comes to our mind: for the forests in $L_{k-1}$, we can define the updated list $L_k = \{ \phi: \phi \in L_{k-1}, e_k \notin \phi  \}$. That is, $L_k$ is the subset of $L_{k-1}$ excluding any forests that contain the edge $e_k$. This method ensures that all forests in $\calF(\calG_{k})$ are equally likely to be included in $L_k$, under the assumption that all forests in $L_{k-1}$ are uniformly sampled from $\calF(\calG_{k-1})$.

However, this method poses a challenge. The number of forests $l_k$ in the updated list $L_k$ will always be less than or equal to $l_{k-1}$. Consequently, if all updates are edge deletions, the size of our sampling forest lists will continually diminish. This reduction in sampling size could result in decreased accuracy when responding to query requests. Therefore, the challenge arises: can we devise an alternative method that maintains the uniformity property without leading to a reduction in the number of sampling forests?

To address the challenge, we define $\Delta \calF(\calG_k)  = \calF(\calG_{k-1}) \setminus \calF(\calG_{k})$. The solution lies in not merely discarding the forests in $\Delta \calF(\calG_k)$ but effectively utilizing them.   Define $\calF(\calG'_{k}) = \{\phi: \phi \in \calF(\calG_{k}), r_{\phi}(u_k) = u_k, r_{\phi}(v_k) \neq u_k\}$. Considering that edge insertion and deletion are inverse operations, Lemma~\ref{le-edge} establishes a bijection between $\Delta \calF(\calG_k)$ and $\calF(\calG'_{k})$. Specifically, for a forest $\phi$ in $\Delta \calF(\calG_k)$, the forest $\phi \setminus \{e_k\}$ is included in $\calF(\calG'_{k})$. This insight provides a method to utilize forests in $\Delta \calF(\calG_k)$ without discarding them. We then introduce an edge deletion update method, the pseudocode of which is described in Algorithm~\ref{alg-deleteedge}.
  
\begin{algorithm}[htbp!]
	\caption{$\textsc{Delete-Update}(\calG_{k-1},L_{k-1},l_{k-1},e_k )$}
	\label{alg-deleteedge}
	\Input{  A digraph $\calG_{k-1} = (V_{k-1},E_{k-1})$,  a list of $l_{k-1}$  spanning converging forest  $L_{k-1}$ uniformly sampled from $\calF(\calG_{k-1})$, an edge $e_k = (u_k,v_k)$ to be deleted	}
	\Output{An updated spanning converging forest list $L_k$}
{
 \textbf{Initialize} :
	$L_k \leftarrow \emptyset$ \\
}

 \For{$\phi$ in $L_{k-1}$}{

\If{$e_k\in \phi$ }{
    $\phi \leftarrow \phi\setminus \{e_k\}$\\
    \textbf{Add} $ \phi $  to   $ L_k  $
 
}\ElseIf{$r_{\phi}(u_k) = u_k$ \& $r_{\phi}(v_k) \neq u_k$}{
    \textbf{Add} $ \phi $  to   $ L_k  $
}\Else{
    \textbf{Add} $ \phi $  to   $ L_k  $ twice
}
}
\textbf{return} $L_k$ \;
\end{algorithm}

\begin{theorem}\label{th-deleteedge}
    For a spanning converging forest list $L_{k-1}$ with $l_{k-1}$ forests uniformly sampled from $\calF(\calG_{k-1})$, and the deletion edge $e_k = (u_k,v_k)$,   all forests in the updated set   $\calF(\calG_{k})$ have the same probability to be included  in  the list returned by Algorithm~\ref{alg-deleteedge}.
\end{theorem}

 The proof of Theorem~\ref{th-deleteedge} is similar to that of  Theorem~\ref{th-addedge}. Theorem~\ref{th-deleteedge} demonstrates that the list $L_k$, generated by Algorithm~\ref{alg-deleteedge}, ensures uniform sampling from the updated set $\calF(\calG_{k})$.
 
 \begin{example}
For a better understanding of Algorithm~\ref{alg-deleteedge}, we present an example  depicted in Figure~\ref{ftoydel}. In this example, we remove the edge $e = (3,1)$ from  graph $\calG_0$, resulting in  graph $\calG_1$. Set $\calF(\calG_0)$  contains 7 spanning converging forests,  shown in the first row of Figure~\ref{ftoydel}, while the set $\calF(\calG_1)$ comprises 4 spanning converging forests,  illustrated in the second row of Figure~\ref{ftoydel}. Assuming a forest list $L_0$ is uniformly sampled from $\calF(\calG_0)$, a straightforward approach might suggest discarding the forests $\phi_5, \phi_6, \phi_7$ in $L_0$ to form $L_1$. However, this strategy will unfortunately reduce the sampling size, as mentioned earlier.

Algorithm~\ref{alg-deleteedge} addresses this challenge   by utilizing the forests $\phi_5, \phi_6, \phi_7$. Lemma~\ref{le-edge} suggests a bijection between $\{ \widehat{\phi}_1,\widehat{\phi}_2,\widehat{\phi}_3\}$ and $\{\phi_5,\phi_6,\phi_7\}$. According to Algorithm~\ref{alg-deleteedge}, for a forest in $L_0$, if it belongs to the set $\{\phi_1, \phi_2, \phi_3\}$, it is directly added to $L_1$. If a forest is part of the set $\{\phi_4, \phi_5, \phi_6\}$, we first remove the edge $e = (3,1)$ and then include the modified forest in $L_1$. This procedure ensures that $\widehat{\phi}_1$ is derived from $\phi_1$ and $\phi_5$, $\widehat{\phi}_2$ from $\phi_2$ and $\phi_6$, and $\widehat{\phi}_3$ from $\phi_3$ and $\phi_7$. To maintain balanced probabilities, Algorithm~\ref{alg-deleteedge} adds $\phi_4$ to $L_1$ twice, which is highlighted with a yellow background, thereby ensuring uniformity in the sampling process from the updated forest set.

 \end{example}

 \begin{figure}[htbp!]
	\centering
	\includegraphics[width=1\columnwidth]{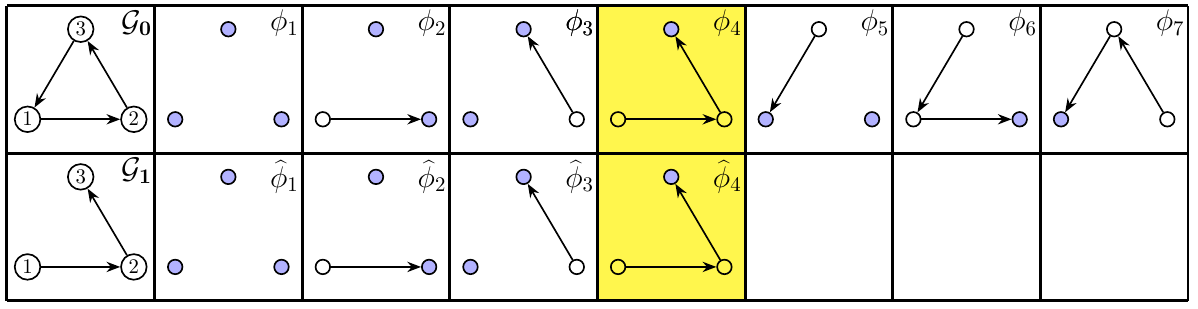}
	\caption{A toy digraph $\calG_0$ and updated graph $\calG_1$ with their spanning converging forests.}\label{ftoydel}	
\end{figure}

\subsection{Prune Technique and Algorithm Details}
In this subsection, we introduce the prune technique   as well as  some details of our algorithms. 

\textbf{Prune Technique.} The time complexity of Algorithm~\ref{alg-addedge} and Algorithm~\ref{alg-deleteedge} is $O(l_{k-1})$, where $l_{k-1}$ is the number of spanning converging forests in the input list $L_{k-1}$. As illustrated previously, the number $l_k$ of spanning converging forests in the updated list $L_{k}$ either equals or exceeds $l_{k-1}$, irrespective of the update being an edge insertion or deletion. Given \(k\) updates, with the number of spanning converging forests incrementally rising as \(l_0 \leq l_1 \leq \ldots \leq l_k\), it becomes essential to mitigate the potential for exponential growth. To this end, we introduce the pruning technique, which involves setting a threshold $l'$, for instance, $l' = 5l_0$.   If \(l_k\) surpasses \(l'\), a pruning action is undertaken to uniformly select $l'$ forests from $L_k$, thereby constituting $L_k'$. This pruning strategy aims to ensure that the time taken for updates and queries remains manageable, avoiding significant increases as the number of updates grows.
 
\textbf{Algorithm Details.} For every spanning converging forest, our algorithms retain  information on the next node for each node in the forest and its root node. This setup results in a space complexity of \(O(ln)\), with \(l\) representing the required number of spanning converging forests. Upon an update, when either Algorithm~\ref{alg-addedge} or Algorithm~\ref{alg-deleteedge} is activated, we avoid  directly replicating  forests to avoid the \(O(n)\) cost in  copying a single forest.  Instead, we only record the updated edge (in the case of insertion) or the adjusted weight (in the case of deletion).  This method, combined with the pruning approach, guarantees that our algorithms achieve  an \(O(1)\) complexity for both query and update operations, enhancing its efficiency and scalability for large-scale network analyses.

\section{Experiments}
In this section, we conduct extensive experiments on various real-life networks in order to evaluate the performance of our algorithms, in terms of accuracy and efficiency.

% Our source code is publicly available on \url{https://anonymous.4open.science/r/Evolving-Graphs-5AB2}.

\subsection{Setup}
 \textbf{Datasets and Equipment.}
 The datasets of selected real networks are publicly available in the KONECT~\cite{Ku13} and SNAP~\cite{LeSo16}. Our experiments are conducted on a diverse range of   undirected and directed networks.  The details of these datasets are presented in Table~\ref{datasets}. All experiments are conducted using the Julia programming language. We conduct all experiments in a computational environment featuring a 2.5 GHz Intel E5-2682v4 CPU with 512GB of primary memory.

% \usepackage{multirow}
\begin{table}[htbp!]\fontsize{8}{11}\caption{Datasets used in experiments. }
\begin{tabular}{cccc}
\hline
{Network} & {Type} & {Nodes} & {Edges} \\\hline
Web-Stanford             & directed              & 281,903                & 2,312,497              \\
Delicious                & undirected            & 536,108                & 1,375,961              \\
Web-Google               & directed              & 875,713                & 5,105,039              \\
Youtube                  & undirected            & 1,134,890              & 2,987,624              \\
Livejournal              & undirected            & 10,690,276             & 112,307,385            \\
Twitter                  & directed              & 41,652,230             & 1,468,365,182          \\ \hline
\end{tabular}\label{datasets}
\end{table}

% %The source code is publicly available on \url{https://github.com/OpinionOptimization/Opinion}.
% \begin{table}[b]\fontsize{8}{11}
% \caption{Datasets used in experiments. }
% \begin{tabular}{cccc}
% \hline
% \textbf{Type}                                                                & \textbf{Dataset}    & $n$        & $m$         \\ \hline
% \multirow{9}{*}{\begin{tabular}[c]{@{}c@{}}undirected\\ graphs\end{tabular}} & web-webbase-2001    & 16,062     & 25,593      \\
% & soc-gemsec-RO       & 41,773     & 125,826     \\& tech-p2p-gnutella   & 62,561     & 147,878     \\& tech-RL-caida       & 190,914    & 607,610     \\& soc-twitter-follows & 404,719    & 713,319     \\& soc-delicious       & 536,108    & 1,375,961   \\& dblp                & 5,624,219  & 12,282,055  \\& livejournal         & 7,489,073  & 112,307,315 \\& delicious           & 33,777,767 & 301,183,342 \\ \hline
% \multirow{9}{*}{\begin{tabular}[c]{@{}c@{}}directed\\ graphs\end{tabular}}  & wikipedialinks        & 17,649     & 296,918     \\ & p2p-gnutella31      & 62,586     & 147,892     \\
% & email-euall         & 265,009    & 418,956     \\& web-Stanford        & 281,903    & 2,312,500   \\& web-Google          & 875,713    & 5,105,039 
%    \\& northwestUSA        & 1,207,945   & 2,820,774    \\& wikitalk            & 2,394,385 
%   & 5,021,410    \\& greatlakes          & 2,758,119 
%   & 6,794,808    \\& fullUSA             & 23,947,347 & 57,708,624 \\ \hline
% \end{tabular}\label{datasets}
% \end{table}
\noindent\textbf{Algorithms and Parameters.} In the evaluation of forest matrix entry queries, we compare our two algorithms \textsc{SFQ} and \textsc{SFQPlus} with the fast linear equation solvers, since direct matrix inversion is computationally infeasible. For undirected graphs, we consider the fast Laplacian solver~\cite{CoKyMiPaJaPeRaXu14}, which is widely used in   computation and optimization problems~\cite{BaZh21,GrAnPrMe21,ZhZh22}. For directed graphs, the fast Laplacian solver no longer applies. Thus, we choose the GMRES algorithm~\cite{SaSc86} to get the ground truth  with a tolerance set to \(10^{-9}\).  

According to Theorem~\ref{th-rf}, we set $\delta = 0.01$, $\epsilon = 0.03$, and the number of spanning converging forest $l$ is given by $l =\left \lceil  (\frac{2}{3\epsilon}+\frac{1}{4\epsilon^2})\log(\frac{2}{\delta})  \right \rceil$. We set the prune threshold to be $l' = 5l$. Given that Wilson's algorithm can be parallelized efficiently, we use 32 computing cores to speed up the process.
% The \textsc{JLT} algorithm combines the Johnson-Lindenstrauss lemma \cite{JoLi84,Ac03} with the  fast Laplacian solver \cite{CoKyMiPaJaPeRaXu14}. 
%  \textsc{JLT} needs  time ofm$O(m\epsilon^{-2}\log^{2.5}n\log\frac{1}{\epsilon}{\rm polyloglog}(n))$ to achieve a relative error bound. 
% On the other hand, the \textsc{UST} algorithm employs a single instance of a Laplacian solver and utilizes Wilson's algorithm to sample uniform spanning trees, thereby deriving the effective resistance. The author in \cite{GrAnPrMe21} elucidates the correlation between the diagonal elements of the forest matrix and the effective resistance in the augmented graph. The \textsc{UST} algorithm incurs a total time complexity of \(\widetilde{O}(m\epsilon^{-2}\log^{3/2}n)\) to guarantee an absolute error of \(\epsilon\) with high probability.

\subsection{Accuracy}
In this subsection, we   evaluate the accuracy of our algorithms \textsc{SFQ} and \textsc{SFQPlus} with the ground truth. For both \textsc{SFQ} and \textsc{SFQPlus}, we consider an undirected graph as a special directed graph,  given that an edge between two nodes in an undirected graph can be considered as two directed edges between the two nodes in a directed graph setting.

\begin{figure}[b!]
	\centering
	\includegraphics[width=1\columnwidth]{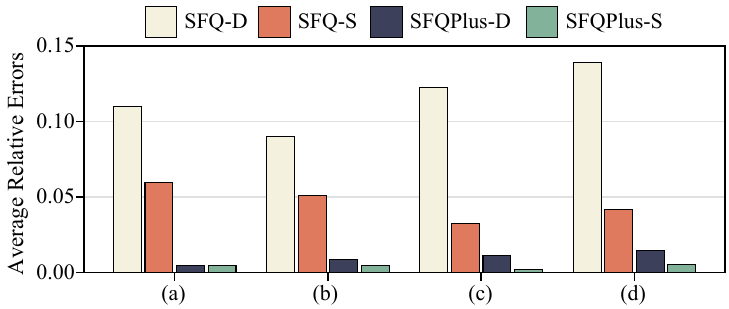}
	\caption{Comparison of  average   relative errors of the diagonals for  algorithms \textsc{SFQ} and \textsc{SFQPlus} on four   graphs: web-Web-Stanford (a), Delicious (b), web-Google    (c), Youtube(d), where suffix  -S  indicates the results on static graphs and  -D  denotes results on the updated graphs.  }\label{f1}	
\end{figure}

Our initial evaluation focuses on the performance of our algorithms in estimating the diagonal of the forest matrix, which has broad applications. %For example, it serves as the forest closeness centrality~\cite{JiBaZh19, GrAnPrMe21} of a network. Besides,  the diagonal of the forest matrix has been closely associated with determinantal point processes in machine learning~\cite{KuTa12}, and  has found relevance through electrical interpretations in multi-agent and network-based problems~\cite{RoFrFa17}.  
In our experiments, for a graph $\mathcal{G}=(V,E)$, we randomly select a node $i\in V$. We then obtain the ground truth by employing the fast Laplacian solver for undirected graphs and GMRES for directed graphs to solve the linear equation $\omega_{ii} = \mathbf{e}_i^\top (\mathbf{I}+\mathbf{L})^{-1}\mathbf{e}_i$. Here we choose four graphs: web-Web-Stanford (a), Delicious (b), web-Google (c), Youtube (d), since the solver was unable to process the two large graphs, Livejournal and Twitter,  constrained by time and memory.  Then we apply algorithms \textsc{SFQ} and \textsc{SFQPlus} to get the estimation values $\widehat{\omega}_{ii}$  and $\widebar{\omega}_{ii}$, respectively.  This procedure is repeated 100 times to calculate the average relative errors. In addition to the static graphs, we explore the scenarios where the  graph evolves   with 50 edges inserted and 50 edges deleted. We repeat the node selection procedure for these   updated graphs and calculate the average relative errors.  The results for these settings are reported in Figure~\ref{f1}.

 Figure~\ref{f1} illustrates that, compared to static graphs, the accuracy of both algorithms experiences varying degrees of decline on updated graphs, confirming our earlier analysis. Specifically, \textsc{SFQ}'s accuracy significantly decreases, rendering its results less reliable, while \textsc{SFQPlus} consistently delivers satisfactory outcomes. This highlights the effectiveness of our variance reduction technique in enhancing accuracy.

% We then conduct experiments on estimating the forest closeness centrality of nodes. The forest closeness centrality was proposed in~\cite{JiBaZh19}, and was shown being more discriminating than many frequently used metrics~\cite{BaZh21}.  %for example, betweenness, harmonic centrality, eigenvector centrality. 
% For any two distinct nodes $i$  and $j$, the forest closeness centrality between  them is defined as $\rho_{ij} = \omega_{ii}+\omega_{jj}-\omega_{ij}-\omega_{ji}$. 

We then conduct experiments to estimate the forest distance  $\rho_{ij}$ between two distinct nodes $i$ and $j$. The forest distance  $\rho_{ij}$ is defined as  $\rho_{ij} = \omega_{ii}+\omega_{jj}-\omega_{ij}-\omega_{ji}$, which was applied in~\cite{JiBaZh19}  to measure the proximity between nodes $i$ and $j$. It has been  shown that the forest distance based node centrality measure is more discriminating than other frequently used metrics for node importance~\cite{BaZh21}. We obtain the ground truth of $\rho_{ij} $ by solving the linear  equations $\rho_{ij} = (\ee_i^\top - \ee_j^\top)(\II+\LL)^{-1}\ee_i +  (\ee_j^\top - \ee_i^\top)(\II+\LL)^{-1}\ee_j$. We then derive the estimations of $\rho_{ij}$ by executing  algorithms  \textsc{SFQ} and  \textsc{SFQPlus} four times   each. We randomly select  400 pairs of distinct nodes and calculate  their forest closeness centrality. The settings for these experiments are consistent with those previously mentioned. The results are displayed in Figure~\ref{f2}. From these results, it is evident that the \textsc{SFQPlus} algorithm achieves better accuracy than \textsc{SFQ} in both static and updated graphs.

\begin{figure}[htbp!]
	\centering
	\includegraphics[width=1\columnwidth]{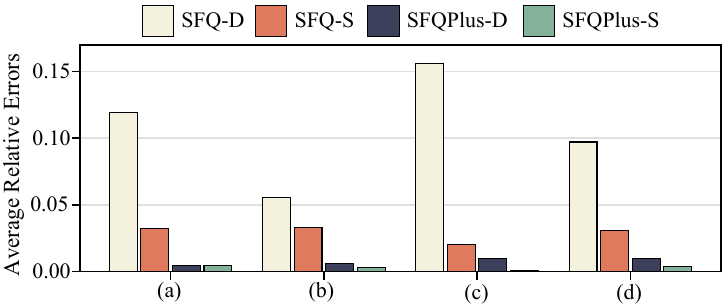}
	\caption{Comparison of  average   relative errors of the forest closeness centrality measures for  algorithms \textsc{SFQ} and \textsc{SFQPlus} on four   graphs: web-Web-Stanford (a), Delicious (b), web-Google    (c), Livejournal(d), where suffix  -S  indicates the results on static graphs and  -D  denotes results on the updated graphs.  }\label{f2}	
\end{figure}

\subsection{Efficiency and Scalability}
 
As illustrated above, the \textsc{SFQPlus} algorithm achieves satisfactory accuracy in comparison to the ground truth. In this subsection, we show the efficiency and scalability of our query algorithms and update strategies on various networks. The results, summarized in Table~\ref{time}, provide a clear indication of these findings. For each network, the \textsc{SFQ} and \textsc{SFQPlus} algorithms are initially executed 100 times, yielding an average time denoted as \textsc{SFQ-S} and \textsc{SFQPlus-S}, respectively,  in the table. Subsequently, the graph undergoes 100 updates, comprising 50 random edge insertions and 50 random edge deletions. The update time is calculated as the average of 100 executions, 50 executions of the \textsc{Insert-Update} algorithm and 50 executions of the \textsc{Delete-Update} algorithm.  After the updates, the \textsc{SFQ} and \textsc{SFQPlus} algorithms are run another 100 times on the updated graph, with the average runtime recorded as \textsc{SFQ-D} and \textsc{SFQPlus-D}. To obtain the ground truth, a fast Laplacian solver for undirected graphs and GMRES for directed graphs are used, with the execution time noted in the Solver column of the table.

From these results, it is observable that under identical conditions, the \textsc{SFQPlus} algorithm typically requires a slightly longer time than \textsc{SFQ}. Moreover, the execution time for both algorithms marginally increases after the updates. Importantly, the time for queries and updates shows insensitivity to the size of the network. This supports the prior analysis that the query and update times are \(O(1)\), significantly less than the time required by the linear solver to determine the ground truth. Remarkably, for two massive networks, Livejournal and Twitter,  both  of which contain over 10 million nodes, the solver fails to run due to time and memory constraints, whereas our algorithms still perform effectively. Our algorithms consistently return results for a single query operation within one second, demonstrating their efficiency and scalability for large-scale network analysis.

\begin{table}[htbp!]\setlength{\tabcolsep}{.8 mm} \fontsize{8}{12}\caption{The running time(seconds)  of the linear solver and    \textsc{SFQ} and \textsc{SFQPlus}  algorithms, as well as the update time, where suffix  -S  indicates the results on static graphs and  -D  denotes results on the updated graphs.   }
\begin{tabular}{ccccccc}
\hline
\multirow{2}{*}{\textbf{Network}} & \multicolumn{6}{c}{\textbf{Running Time(seconds)}}       \\ \cline{2-7} 
                                  & \textsc{SFQ-S}  & \textsc{SFQ-D} & \textsc{SFQPlus-S} & \textsc{SFQPlus-D} & Update & Solver \\ \hline
Stanford                      & 0.0008 & 0.027 & 0.0013    & 0.028     & 0.097  & 22.17  \\
Delicious                         & 0.0008 & 0.097 & 0.0012    & 0.098     & 0.429  & 50.06  \\
Google                        & 0.0007 & 0.153 & 0.0013    & 0.157     & 0.485  & 81.61  \\
Youtube                           & 0.0009 & 0.252 & 0.0014    & 0.256     & 0.885  & 255.64 \\
Livejournal                       & 0.0009 & 0.314 & 0.0014    & 0.362     & 1.121  & -      \\
Twitter                           & 0.0011 & 0.421 & 0.0021    & 0.489     & 1.893  & -      \\ \hline
\end{tabular}\label{time}
\end{table}

\section{Conclusions}
In this paper, we addressed the problem of efficiently querying the entries of the forest matrix of a dynamically evolving graph.  Leveraging an extension of Wilson's algorithm, we presented an  algorithm \textsc{SFQ} as our foundational approach. We further enhanced this foundation and developed \textsc{SFQPlus}, an advanced algorithm that incorporates an innovative variance reduction technique to improve the accuracy of estimations for the  entries of forest matrix.  Additionally, we proposed innovative forest updating techniques to manage evolving graphs, including edge additions and deletions. Our methods maintains $O(1)$ time complexity for both updates and queries, while ensuring unbiased estimates of the entries  of the forest matrix entries. Extensive experimentation on various real-world networks validates the effectiveness and efficiency of our algorithms.  Moreover, our algorithms are scalable to massive graphs with over forty million nodes.

In future work, we plan to extend our algorithms to other problems on evolving graphs, such as dynamically solving linear systems associated with the forest matrix, thereby broadening their applicability in network analysis.

\section*{Acknowledgements}
This work was supported by the National Natural Science Foundation of China (Nos. 62372112, U20B2051, and 61872093).

\bibliographystyle{ACM-Reference-Format}
\balance
\bibliography{main}

\clearpage

\end{document}